\documentclass{article}

\usepackage{tikz}
\usetikzlibrary{matrix,arrows.meta,positioning,shapes.geometric}

\usepackage{lineno,hyperref}
\usepackage{amsthm}
\usepackage{fullpage}
\usepackage{amsmath}
\usepackage{amsfonts}
\usepackage{amssymb}
\usepackage{graphicx}
\usepackage{color}
\usepackage{epsf}
\usepackage{float}
\usepackage{ulem}
\usepackage{breqn}
\usepackage{bigints}
\usepackage[english]{babel}
\usepackage[nottoc]{tocbibind}

\newtheorem{theorem}{Theorem}[section] 

\theoremstyle{definition}

\theoremstyle{remark}
\newtheorem{remark}[theorem]{Remark}

\title{Drivers of Transient Dynamics and Persistence in Dengue: Insights from Sensitivity and Stochastic Modeling}
\author{Cesar Alberto Rosales-Alcantar~\footnote{Corresponding author}~\thanks{Centro de Investigaci\'on en Matem\'aticas A.C. Unidad M\'erida, km 5.5 Carretera Sierra Papacal - Chuburn\'a Puerto, M\'erida 97302, Yucat\'an, México},  \and Marcos A. Capistrán~\thanks{Centro de Investigaci\'on en Matem\'aticas A.C. Unidad M\'erida, km 5.5 Carretera Sierra Papacal - Chuburn\'a Puerto, M\'erida 97302, Yucat\'an, México}}
\date{February 2, 2026}

\begin{document}

\maketitle
\begin{abstract}
We investigate how key epidemiological parameters shape both seasonal epidemics and the persistence of dengue transmission. Our findings confirm known mechanistic drivers of epidemic variability and introduce a ranking of parameter importance in our dengue model, which in turn informs the prioritization of public health policies. We propose a stochastic vector–host model with waning immunity, exogenous infection, and vertical transmission. To assess parameter influence, we first qualitatively analyze the macroscopic model. We then perform a multivariate Sobol sensitivity analysis of epidemic summary statistics, and examine the variance of the endemic equilibrium as a function of model parameters. We show that the macroscopic model is well posed, vertical transmission lowers the threshold for persistence, and low spatial coupling increases infectious endemic equilibria. The vector–host population ratio and host recovery rate have the largest first-order and total sensitivity indices, surpassing the contact rates; this implies that control measures during seasonal dengue should prioritize protecting infectious hosts from mosquito bites. Finally, we show that the covariance of hosts and vectors at the endemic equilibrium is asynchronous in the contact-rate plane. This robust pattern has epidemiological, ecological and evolutive interpretations. A dengue strain has two niches to exploit during the endemic regime, and coexisting strain have two niches each. Moreover, large fluctuations in a given strain during the endemic regime provide a mechanistic explanation for high vertical transmission, enabling viral reservoirs that can hatch and trigger outbreaks in the following season.
We argue that our model and results can be adapted to address specific public health questions to guide dengue control using field data.
\end{abstract}
%
\pagestyle{plain}

\section{Introduction}
\label{sec:introduction}

In this paper, we build on the framework introduced by Dietz~\cite{Dietz1975} to investigate how key parameters influence dengue transmission dynamics. Specifically, we examine the effects of contact rates, exogenous infection, vertical transmission, waning immunity, and the relative sizes of host and vector populations on the scale and progression of dengue epidemics. Our analysis focuses on the impact of these parameters: (i) on seasonal epidemic summary statistics, and (ii) the variance of the endemic equilibrium. The rationale for using a one-strain model is that it avoids the added complexity of cross-immunity and antibody-dependent enhancement, allowing a clearer analysis of fundamental mechanisms such as the basic reproduction number, persistence conditions, and the effects of vector-control measures. This formulation also portrays the type of surveillance data typically available, which often lack serotype resolution, and enables an assessment of how demographic, entomological, and epidemiological parameters shape both seasonal epidemic dynamics and long-term endemic levels. The results provide insight into the main drivers of transmission and support the design of effective control strategies, while also serving as a baseline for more detailed multi-strain models in the future.

\textbf{Related work.} Beginning with the seminal contributions of Ross~\cite{ross1915some,ross1916application} and MacDonald~\cite{macdonald1952analysis,macdonald1957epidemiology}, many classical and contemporary studies have examined the dynamics of vector-borne diseases. Regarding dengue modeling, Dietz~\cite{Dietz1975} introduced a deterministic host–vector model, defined $\mathcal{R}_{0}$, and laid the foundation for modern vector control and vaccination strategies. Esteva and Vargas~\cite{esteva1998analysis} developed a rigorous mathematical framework demonstrating the global stability of the endemic equilibrium and showing how vector control measures governed by threshold conditions can effectively influence disease dynamics. Adams and Boots~\cite{adams2010important} showed that vertical transmission in mosquitoes, although not the sole factor, plays an important role in dengue virus persistence, particularly during inter-epidemic periods. Lloyd \textit{et al.}~\cite{lloyd2007stochasticity} analyzed the role of stochasticity and heterogeneity in vector–host models, while Colon \textit{et al.}~\cite{colon2013effects} examined the effects of weather and climate change on dengue epidemics. Alonso \textit{et al.}~\cite{alonso2007stochastic} developed a stochastic theory showing how random nonlinear epidemic dynamics and weak spatial coupling, can amplify noise and drive the transition from regular to irregular epidemic cycles. Their work demonstrated that demographic stochasticity is a key mechanism behind the complex oscillatory patterns observed in childhood diseases and that these effects remain robust even under seasonal forcing. Contemporary reviews of vector-borne epidemic modeling include Natsir \textit{et al.}~\cite{natsir2023transmission}, who showed that all four dengue serotypes remain endemic in Asia. Smith \textit{et al.}~\cite{smith2012ross,smith2014recasting} provided a comprehensive synthesis of Ross–Macdonald theory, establishing it as a mathematical framework for analyzing the dynamics and control of mosquito-borne pathogens. Their work emphasized the assumptions underlying mosquito–host transmission and highlighted the central role of vectorial capacity and reproduction numbers in quantifying transmission potential and guiding control strategies. Aguiar \textit{et al.}~\cite{aguiar2022mathematical} reviewed mathematical models for dengue epidemiology, with emphasis on multi-strain frameworks incorporating host-to-host and vector–host transmission, antibody-dependent enhancement, and temporary cross-immunity. Their contribution showed how these mechanisms explain the complex and sometimes chaotic dynamics of dengue transmission and provide insights to inform public health strategies.

\textbf{Contributions and limitations.}  
This work is intended as a first--principles, coarse--grained analysis of dengue transmission, aimed at identifying the relative importance of key epidemiological parameters and stochastic mechanisms, rather than providing precise operational parameter estimates for specific surveillance settings.

First, we show that the stability of the disease-free and endemic equilibria in our model is determined by the basic reproductive number, with $\mathcal{R}_{0}=1$ as the threshold. We also show that endemic equilibrium remains stable under small perturbations of the force of infection representing weak spatial coupling or exogenous infections, and the corresponding infectious steady states become larger. Next, using parameter ranges reported by Aguiar \textit{et al.}~\cite{aguiar2022mathematical} and others, we assess the influence of different parameters on transient, or seasonal, dengue dynamics. We conduct a multivariate Sobol sensitivity analysis, which provides both first-order and total-effect indices to quantify the relative contribution of each parameter to seasonal epidemic peak, trough, variability and total population burden. For the system initial conditions in the sensitivity analysis, we assume that a small number of infectious hosts and mosquitoes is sufficient to initiate local community transmission in a naive population at the beginning of the season. Generalizations of the experiment, such as considering weather driven mosquito population increase or a specific attack rate, are natural extensions that should follow from public health questions.

The sensitivity analysis confirms that the vector--host population ratio $C$ and the human recovery rate $\gamma_h$ dominate the system’s transient dynamics sensitivity, consistently yielding the largest indices across epidemic summary statistics such as peak, trough, variability, and total infectious burden of each species. In contrast, the contact rates $\beta_h$ and $\beta_v$, although biologically central to transmission, rank lower than the scaling effects of $C$ and the temporal modulation induced by $\gamma_h$. This quantitative ranking not only recovers known mechanistic insights—that population scaling and infectious period are principal drivers of seasonal dengue—but also establishes a hierarchy among secondary parameters. This ordering has direct policy implications: for example, reducing vector density or shortening the infectious period is predicted to be more effective than targeting contact intensity alone.  

Finally, similar to Alonso \textit{et al.}~\cite{alonso2007stochastic}, we use the endemic equilibrium, the linear noise approximation, the Wiener–Khinchin theorem, and a Lyapunov algebraic equation to derive an expression for the variance of the endemic equilibrium as a function of the model parameters. Numerical evidence shows that the regions of maximum variance of hosts and vectors are asynchronous in the host-to-vector and vector-to-host contact-rate plane, and that this structure is robust to small parameter perturbations, changing only the variance maximum values. Gillespie realizations in this setting further illustrate how selected values of the contact rates can give rise to a significant number of infections around the endemic equilibrium over the course of a year. These findings suggest that further exploration of the endemic state variance as a quantitative measure of fluctuations of cases past the seasonal outbreak is warranted to inform public health policies aimed at reducing transmission risk.  

A limitation of our study is that it relies on published parameter ranges, which may not fully capture the variability across different epidemiological settings. Consequently, future analyses should be guided by specific public health questions and accompanied by Bayesian inference of system parameters.  

The remainder of the paper is organized as follows. Section~\ref{sec:methods} introduces the transmission model and provides a brief overview of the mathematical techniques employed in the analysis. Section~\ref{sec:results} presents the main findings, with particular emphasis on their implications for the transient and persistent dynamics of dengue epidemics. Section~\ref{sec:discussion} places these results in a broader epidemiological context, highlighting their relevance for understanding dengue persistence and informing control strategies. Finally, Section~\ref{sec:conclusions} summarizes the key insights of the study and outlines potential directions for future research.

\section{Methods}
\label{sec:methods}

In this section, we introduce a mathematical model for dengue transmission that includes key epidemiological features of the disease, including human-mosquito interactions and vector vertical transmission. We also outline a set of analytical and computational tools used to investigate the mechanisms that sustain endemic dengue transmission in human populations. These tools allow us to explore both deterministic and stochastic aspects of the model, and to identify the dominant drivers of seasonality and persistence in disease dynamics.

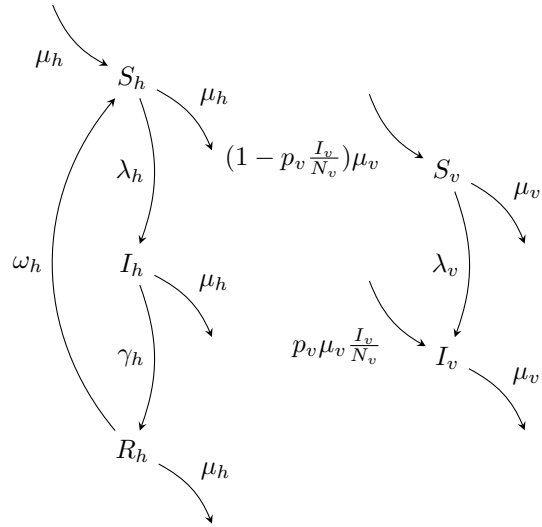
\begin{figure}
\begin{center}
\begin{tikzpicture}
  \matrix (m) [matrix of math nodes,row sep=2em,column sep=2em]
  { \vphantom{nil} & & & & \vphantom{nil} & & \\
      & S_{h} & & \vphantom{nil} & \vphantom{nil} & \vphantom{nil} &\\
     & & \vphantom{nil} & & \vphantom{nil} & S_{v} & \vphantom{nil}\\  
      & I_{h} & & \vphantom{nil} & \vphantom{nil} & \vphantom{nil} & \vphantom{nil}\\     
      & &\vphantom{nil} & & \vphantom{nil}& I_{v} & \vphantom{nil}\\ 
      & R_{h} & & \vphantom{nil} && \vphantom{nil}& \vphantom{nil}\\
      & &\vphantom{nil} & & \vphantom{nil}& & \vphantom{nil}\\};
  \path[-stealth]
    (m-1-1) [bend right = 20]edge node [below left] {$\mu_{h} $} (m-2-2)
    (m-2-5) [bend right = 20]edge node [below left] {$(1-p_v\frac{I_v}{N_v})\mu_{v} $} (m-3-6) 
	(m-2-2) [bend left = 20] edge node [left] {$\lambda_h$} (m-4-2)    
    (m-2-2) [bend left = 20] edge node [above right] {$\mu_{h}$} (m-3-3)
	(m-4-5) [bend right = 20] edge node [below left] {$p_v\mu_{v}\frac{I_v}{N_v}$} (m-5-6)    
    (m-3-6) [bend left = 20] edge node [above right] {$\mu_{v}$} (m-4-7)    
	(m-3-6) [bend left = 20] edge node [left] {$\lambda_v$} (m-5-6)    
    (m-4-2) [bend left = 20] edge node [above right] {$\mu_{h}$} (m-5-3)
    (m-4-2) [bend left = 20] edge node [left] {$\gamma_{h}$} (m-6-2)          
    (m-6-2) [bend left = 40] edge node [left] {$\omega_{h}$} (m-2-2)              
    (m-5-6) [bend left = 20] edge node [above right] {$\mu_{v}$} (m-6-7)  
	(m-6-2) [bend left = 20] edge node [above right] {$\mu_{h}$} (m-7-3);	 
\end{tikzpicture}
\end{center}
\caption{{\bf Dengue transmission model with vertical transmission in vectors.} 
The diagram shows a compartmental model of dengue transmission between humans and mosquitoes. Humans are divided into susceptible ($S_h$), infectious ($I_h$), and recovered ($R_h$) classes. Mosquitoes are divided into susceptible ($S_v$) and infectious ($I_v$) classes.
Susceptible humans are infected at rate $\lambda_h$ through bites from infectious mosquitoes. They recover at rate $\gamma_h$, gain temporary immunity, and return to the susceptible class at rate $\omega_h$. Human mortality occurs at rate $\mu_h$.
Mosquitoes are born at a baseline rate $\mu_v$, with a fraction $p_v \tfrac{I_v}{N_v}$ infected by vertical transmission. Susceptible mosquitoes become infected at rate $\lambda_v$ through biting infectious humans and die at rate $\mu_v$.
Human and mosquito populations are constant, with $N_h = S_h + I_h + R_h$ and $N_v = S_v + I_v$. The parameter $\eta$ (not shown) represents spatial coupling in hosts.
}
\label{fig:model}
\end{figure}

\begin{table}[htp]
\begin{center}
\textbf{Human events}
\begin{displaymath}
\renewcommand{\arraystretch}{1.4}
\begin{array}{|l|l|l|}
\hline
\mbox{Birth} & \alpha_{1}(x)=\mu_{h} N_h & \bar{v}_{1}=(+1,0,0,0,0)\\[0.3ex]
\mbox{Infection} & \alpha_{2}(x)=\lambda_h x_1& \bar{v}_{2}=(-1,+1,0,0,0) \\[0.3ex]
\mbox{Removal} & \alpha_{3}(x)=\gamma_h x_2& \bar{v}_{3}=(0,-1,+1,0,0) \\[0.3ex]
\mbox{Immunity loss} & \alpha_{4}(x)=\omega_h x_3& \bar{v}_{4}=(+1,0,-1,0,0) \\[0.3ex]
\mbox{Susceptible death} & \alpha_{5}(x)=\mu_h x_1& \bar{v}_{5}=(-1,0,0,0,0) \\[0.3ex]
\mbox{Infectious death} & \alpha_{6}(x)=\mu_h x_2& \bar{v}_{6}=(0,-1,0,0,0) \\[0.3ex]
\mbox{Removed death} & \alpha_{7}(x)=\mu_h x_3& \bar{v}_{7}=(0,0,-1,0,0) \\[0.3ex]
\hline
\end{array}
\end{displaymath}

\textbf{Vector events}
\begin{displaymath}
\renewcommand{\arraystretch}{1.4}
\begin{array}{|l|l|l|}
\hline
\mbox{Susceptible birth} & \alpha_{8}(x)= \mu_v(N_v-p_vx_5)& \bar{v}_{8}=(0,0,0,+1,0)\\[0.3ex]
\mbox{Infectious birth}   & \alpha_{9}(x)= \mu_vp_vx_5  & \bar{v}_{9}=(0,0,0,0,+1)\\[0.3ex]
\mbox{Infection} & \alpha_{10}(x)= \lambda_v x_4 & \bar{v}_{10}=(0,0,0,-1,+1)\\[0.3ex]
\mbox{Susceptible death}    & \alpha_{11}(x)= \mu_{v} x_4 & \bar{v}_{11}=(0,0,0,-1,0)\\[0.3ex]
\mbox{Infectioous death}    & \alpha_{12}(x)= \mu_{v} x_5 & \bar{v}_{12}=(0,0,0,0,-1)\\[0.3ex]
\hline
\end{array}
\end{displaymath}
\end{center}
\caption{{\bf Epidemic events, propensities, and stoichiometry of the stochastic dengue transmission model.} 
This table describes the discrete events that define dengue transmission between humans and mosquitoes, modeled as a continuous-time Markov jump process. For each event, the propensity $\alpha_i(x)$ gives the rate of occurrence, and the stoichiometric vector $\bar{v}_i$ specifies how the system state $x = (x_1, x_2, x_3, x_4, x_5)$ changes.
For humans, the events are birth, infection, recovery, loss of immunity, and death. Individuals are born susceptible, infected at rate $\lambda_h = \beta_h \tfrac{x_5}{N_h} + \epsilon_h$, recover at rate $\gamma_h$, and lose immunity at rate $\omega_h$. All humans die at rate $\mu_h$, while the total population size remains constant.
For mosquitoes, susceptible individuals are born at a rate proportional to the uninfected fraction, and infectious individuals arise through vertical transmission at rate $\mu_v p_v x_5$. Mosquitoes are infected at rate $\lambda_v = \beta_v \tfrac{x_2}{N_v}$ and die at rate $\mu_v$.
All propensities follow mass-action kinetics. Population sizes are $N_h = x_1 + x_2 + x_3$ and $N_v = x_4 + x_5$, both constant. The vertical transmission fraction is $p_v \in [0,1]$.
}
\label{tab:model}
\end{table}

\begin{table}[ht]
\begin{center}
\renewcommand{\arraystretch}{1.4} 
\begin{tabular}{|p{5.5cm}|c|p{3.0cm}|p{1.15cm}|}
\hline
\multicolumn{4}{|c|}{\textbf{HUMAN PARAMETERS}} \\[0.5ex]
\hline
\textbf{Definition} & \textbf{Symbol} & \textbf{Value} & \textbf{Source} \\[0.5ex]
\hline
Population size & $N_h$ & $[10^2,10^5]$ persons & \cite{Gomez-Vargas2024} \\[0.3ex]
Birth/death rate & $\mu_h$ & $[\frac{1}{70},\frac{1}{50}]\text{years}^{-1}$ & \cite{aguiar2022mathematical} \\[0.3ex]
Recovery rate & $\gamma_h$ & $[\frac{1}{12},\frac{1}{3}]\text{days}^{-1}$ & \cite{aguiar2022mathematical} \\[0.3ex]
Immunity waning rate & $\omega_h$ & $[\frac{1}{2},2]\text{years}^{-1}$ & \cite{sanchez2023mathematical} \\[0.3ex]
Exogenous infection rate & $\epsilon_h$ & $10^{-5}\text{days}^{-1}$ & \cite{alonso2007stochastic} \\[0.3ex]
\hline
\multicolumn{4}{|c|}{\textbf{MOSQUITO PARAMETERS}} \\[0.5ex]
\hline
Population size & $N_v$ & $[10^{2},10^{9}]$ vectors & \cite{Gomez-Vargas2024} \\[0.3ex]
Birth/death rate & $\mu_v$ & $[\frac{1}{20},\frac{1}{7.7}]\text{days}^{-1}$ & \cite{aguiar2022mathematical} \\[0.3ex]
Vertical transmission probability & $p_v$ & $[0,0.5]$ & \cite{mbaoma2025significance} \\[0.3ex]
To-human transmission rate & $\beta_h$ & $[0,\frac{1}{4}]\text{days}^{-1}$ & \cite{aguiar2022mathematical} \\[0.3ex]
From-human transmission rate & $\beta_v$ & $[0,\frac{3}{4}]\text{days}^{-1}$ & \cite{aguiar2022mathematical} \\[0.3ex]
\hline
\end{tabular}
\end{center}
\caption{\textbf{Model parameters.} Human and mosquito parameters of the vector–host dengue model. The table shows each parameter’s definition, symbol, range of values, and source. Human parameters include demographic rates, recovery, immunity waning, and exogenous infection. Mosquito parameters include demographic rates, vertical transmission, and transmission between humans and vectors. All parameter ranges are taken from the cited references.}
\label{tab:parameters}
\end{table}


\subsection{Qualitative analysis of the macroscopic model}
\label{sec:macroscopic}
In this section we carry out the qualitative analysis of the macroscopic model defined by Figure~\ref{fig:model} and 
Table~\ref{tab:model}. Using the standard notation for deterministic models we have
\begin{equation}
\begin{array}{ccl}
\dot{S}_h &=& \mu_h N_h - \beta_{h} \frac{I_v}{N_h}S_h - \epsilon_hS_h - \mu_h S_h + \omega_hR_h,\\
\dot{I}_h &=& \beta_{h} \frac{I_v}{N_h}S_h + \epsilon_hS_h - (\gamma_h+\mu_h) I_h,\\
\dot{R}_h &=& \gamma_h I_h - (\omega_h+\mu_h) R_h,\\
\dot{S}_v &=& \mu_v(N_v-p_v I_v) - \beta_{v}\frac{I_h}{N_h}S_v  - \mu_v S_v,\\
\dot{I}_v &=& \beta_{v} \frac{I_h}{N_h}S_v  + (p_v-1) \mu_v I_v. \\
\end{array}
\label{eqn:vectorhost}
\end{equation}

Here, the force of infection on the human host is given by $\lambda_h = \beta_h \frac{I_v}{N_h} + \epsilon_h$, where the first term represents vector-mediated transmission and the second term $\epsilon_h$ accounts for a background infection rate due to weak spatial coupling, or exogenous infections. The force of infection on the vector population is given by $\lambda_v = \beta_v \frac{I_h}{N_v}$. Host individuals are born and die at rate $\mu_h$, recover at rate $\gamma_h$, and lose immunity at rate $\omega_h$. For the mosquito, we assume equal birth and death rates $\mu_v$, and that a fraction $p_v$ of newly born mosquitoes is infected vertically in proportion to the current number of infectious mosquitoes.

To analyze model~\eqref{eqn:vectorhost}, we compute first the basic reproduction number $\mathcal{R}_0$ as the dominant eigenvalue of the next-generation matrix~\cite{van2002reproduction,heesterbeek2002brief} provided system~\eqref{eqn:vectorhost} satisfies the conditions required by van den Driessche and Watmough theory. In particular, any solution starting in the nonnegative orthant remains in this region for all future times. The matrices representing the appearance of new infections and the transfer of individuals out of the infectious compartments are, respectively,
\[
\mathcal{F}=
\begin{pmatrix}
\beta_h \frac{I_v}{N_h} S_h +\epsilon_hS_h\\
\beta_v \frac{I_h}{N_h} S_v + p_v \mu_v I_v
\end{pmatrix},
\quad
\mathcal{V}=
\begin{pmatrix}
(\gamma_h + \mu_h) I_h \\
\mu_v I_v
\end{pmatrix}.
\]

The Jacobians of $\mathcal{F}$ and $\mathcal{V}$ with respect to the infectious compartments $I_h$ and $I_v$, evaluated at the disease-free equilibrium $x_{\text{DFE}} = (N_h, 0, 0, N_v, 0)$, are

\begin{displaymath}
F = \begin{pmatrix}0 & \beta_h\\ \beta_v \frac{N_v}{N_h} & p_v\mu_v\end{pmatrix},
\quad
V = \begin{pmatrix} \gamma_h+\mu_h & 0\\ 0 & \mu_v\end{pmatrix}.
\end{displaymath}

The basic reproduction number of system~\eqref{eqn:vectorhost} is

\begin{equation}
\label{eqn:R0}
\mathcal{R}_{0} = \rho(FV^{-1}) = \frac{p_v}{2} + \sqrt{\left(\frac{p_v}{2}\right)^2 + \mathcal{R}^{2}},
\end{equation}

where $\mathcal{R}$ denotes the reproduction number without vertical transmission, defined as

\begin{equation}
\mathcal{R} = \sqrt{\frac{\beta_h \beta_v}{(\gamma_h + \mu_h)\mu_v}\frac{N_v}{N_h}}.
\end{equation}

Thus, $\mathcal{R}_0$ may exceed one even if $\mathcal{R} < 1$ when vertical transmission is sufficiently strong. This highlights that dengue control requires addressing multiple transmission routes, a point we explore further in the paper.

Equation~\eqref{eqn:R0} agrees with Adams and Boots~\cite{adams2010important} and related work. For convenience, we can also express equation~\eqref{eqn:R0} as

\begin{equation}
\label{eqn:R}
\mathcal{R} = \sqrt{\mathcal{R}_{0}(\mathcal{R}_{0} - p_v)}.
\end{equation}

\begin{theorem}[Stability of the Disease-Free Equilibrium]
Consider an infectious disease model~\eqref{eqn:vectorhost} with a basic reproduction number~\eqref{eqn:R0}. Let $x_{DFE} = (N_h, 0, 0, N_v, 0)$ denote the disease-free equilibrium. Then:
\begin{itemize}
\item If $\mathcal{R}_0<1$, the disease-free equilibrium $x_{DFE}$ is locally asymptotically stable.
\item If $\mathcal{R}_0>1$, the disease-free equilibrium $x_{DFE}$ is unstable.
\end{itemize}
\end{theorem}

\begin{proof}
Let $J_{DFE}$ denote the Jacobian matrix of the right-hand side of system~\eqref{eqn:vectorhost}, evaluated at the disease-free equilibrium. The characteristic polynomial is given by
\begin{equation}
\label{eqn:char_pol}
p_{DFE}(\lambda) = \det(J_{DFE} - \lambda I) = -(\mu_h + \lambda)(\omega_h + \mu_h + \lambda)(\mu_v + \lambda)\,g(\lambda),
\end{equation}
where
\begin{displaymath}
\begin{split}
g(\lambda) &= \lambda^2 + \big[\gamma_h + \mu_h + (1 - p_v)\mu_v\big] \lambda + (\gamma_h + \mu_h)\mu_v(1 - p_v) - \beta_h \beta_v \frac{N_v}{N_h} \\
&= \lambda^2 - T\lambda + D,
\end{split}
\end{displaymath}
with $T = -\left[\gamma_h + \mu_h + (1 - p_v)\mu_v\right] < 0$ for all parameter values.

Therefore, the disease-free equilibrium $x_{DFE}$ is locally asymptotically stable if and only if $D > 0$, which holds precisely when
\[
1 - p_v > \mathcal{R}^2 = \mathcal{R}_0(\mathcal{R}_0 - p_v),
\]
which is equivalent to $\mathcal{R}_0 < 1$. Conversely, if $\mathcal{R}_0 > 1$, then $D < 0$, and the disease-free equilibrium is unstable.
\end{proof}

Let us respectively denote 

\begin{equation}
\label{eqn:quotients}
\mathcal{R}_v=\frac{\beta_v}{\mu_v},\quad Q=\frac{\gamma_h}{\omega_h+\mu_h},\quad C=\frac{N_v}{N_h}
\end{equation}

the vector reproductive number, the immunity residence quotient, and the population ratio. Notice that $Q>1$ means the population is gaining more immune individuals than is losing, while $Q<1$ means individuals exit the immune class faster than they enter it. The system~\eqref{eqn:vectorhost} has and endemic equilibrium $x_{EE}=(S_h^{EE},I_h^{EE},R_h^{EE},S_v^{EE},I_v^{EE})$ given by

\begin{equation}
\label{eqn:ee}
\begin{split}
S_{h}^{EE}&=N_h\left(1-(1+Q)\frac{\mathcal{R}^{2}-(1-p_v)}{\mathcal{R}^{2}(1+Q)+\mathcal{R}_v}\right)\\
I_{h}^{EE}&=N_h\frac{\mathcal{R}^{2}-(1-p_v)}{\mathcal{R}^{2}(1+Q)+\mathcal{R}_v}\\
R_{h}^{EE}&=N_h\frac{\mathcal{R}^{2}-(1-p_v)}{\mathcal{R}^{2}(1+Q)+\mathcal{R}_v}Q\\
S_{v}^{EE}&=N_v\frac{(1-p_v)(\mathcal{R}^{2}(1+Q)+\mathcal{R}_v)}{\mathcal{R}_v(\mathcal{R}^{2}-(1-p_v))+(1-p_v)(\mathcal{R}^{2}(1+Q)+\mathcal{R}_v)}\\
I_{v}^{EE}&=N_v\frac{\mathcal{R}_v(\mathcal{R}^{2}-(1-p_v))}{\mathcal{R}_v(\mathcal{R}^{2}-(1-p_v))+(1-p_v)(\mathcal{R}^{2}(1+Q)+\mathcal{R}_v)}\\
\end{split}
\end{equation}

It is apparent that the existence of the endemic equilibrium \eqref{eqn:ee} requires $\mathcal{R}_{0} > 1$. Also, model~\eqref{eqn:vectorhost} has no backward bifurcation since the derivative of the infectious populations equilibria with respect to the basic reproductive number evaluated at $\mathcal{R}_{0}=1$ is positive:
 
\begin{displaymath}
\begin{split}
\left.\frac{dI_{h}^{EE}}{d\mathcal{R}_{0}}\right|_{\mathcal{R}_{0}=1} &= N_{h}\frac{2-p_v}{(1-p_v)(1+Q)+\mathcal{R}_v}>0,\\
\left.\frac{dI_{v}^{EE}}{d\mathcal{R}_{0}}\right|_{\mathcal{R}_{0}=1} &= \frac{\mathcal{R}_{v}C(1-p_v)(2-p_v)}{(1-p_v)(1+Q)+\mathcal{R}_v}>0.\\
\end{split}
\end{displaymath}

On the other hand, the endemic prevalence increases with the square of the reproduction number $\mathcal{R}$, implying that increases in the contact rates $\beta_h$, $\beta_v$, and the population ratio $C = \frac{N_v}{N_h}$ can lead to higher infection prevalence in both populations. An effective strategy to reduce the reproduction number is to lower $\beta_v$ by protecting infectious hosts from mosquito bites. On the other hand, the vertical transmission fraction $p_v$ lowers the threshold for endemicity. As $p_v$ increases, the disease can persist at lower values of $\mathcal{R}$. Ecologically, this implies that mosquito populations can act as reservoirs, sustaining transmission even when host-to-vector cycles are weak. This highlights the importance of targeting vertical transmission in control strategies. Finally, the immunity residence time $Q$, which reflects the balance between recovery and loss of immunity (or death), directly influences both $S_{h}^{EE}$ and $R_{h}^{EE}$. Larger values of $Q$ (longer-lasting immunity) result in higher equilibrium levels of recovered individuals and lower levels of susceptibles. This indicates that longer immune memory reduces the pool of susceptible hosts and lowers the disease burden, even when transmission persists. We conclude that vaccination or other strategies that extend the duration of immunity could be effective in reducing endemic prevalence. Next, we show that $\mathcal{R}_{0} > 1$ is also a condition for the endemic equilibrium local stability.

\begin{theorem}[Global Stability of the Endemic Equilibrium]
\label{theo:ee}
Consider system~\eqref{eqn:vectorhost}, and let \( x_{\mathrm{EE}} \) denote the endemic equilibrium given in~\eqref{eqn:ee}. Suppose the basic reproduction number \(\mathcal{R}_0\) is defined as in~\eqref{eqn:R0}. Then, the endemic equilibrium \( x_{\mathrm{EE}} \) is globally asymptotically stable in the interior of the feasible region whenever \( \mathcal{R}_0 > 1 \).
\end{theorem}

\begin{proof}
Define 
\[
s_1 = \omega_h + \mu_h, \quad 
s_2 = (1-p_v)\mu_v + \beta_v \frac{I_h^{EE}}{N_h}, \quad 
s_3 = \gamma_h + \mu_h + \beta_h \frac{I_v^{EE}}{N_h},
\]
\[
\Delta_1 = \left(\beta_h \frac{S_h^{EE}}{N_h}\right)\left(\beta_v\frac{S_v^{EE}}{N_h}\right), 
\quad 
\Delta_2 = \gamma_h \beta_h \frac{I_v^{EE}}{N_h}.
\]
Let \(J_{EE}\) be the Jacobian matrix of the right-hand side of system~\eqref{eqn:vectorhost} evaluated at the endemic equilibrium. The corresponding characteristic polynomial is
\begin{equation*}
p_{EE}(\lambda) = \det(J_{EE} - \lambda I) = -(\lambda + \mu_v)(\lambda + \mu_h)(\lambda^3 + A\lambda^2 + B\lambda + C),
\end{equation*}
where
\begin{eqnarray*}
A & = & s_1 + s_2 + s_3, \\
B & = & s_1(s_2 + s_3) + (s_2 s_3 - \Delta_1) + \Delta_2, \\
C & = & s_1(s_2 s_3 - \Delta_1) + s_2 \Delta_2.
\end{eqnarray*}

We claim that \(A\), \(B\), and \(C\) are positive whenever \(s_2 s_3 - \Delta_1 > 0\). This implies that the polynomial \(p_{EE}(\lambda)\) has no positive real roots. Moreover, there exist positive parameters \(r, T, D\) such that
\begin{equation*}
\lambda^3 + A\lambda^2 + B\lambda + C = (\lambda + r)(\lambda^2 + T\lambda + D).
\end{equation*}
Therefore, the endemic equilibrium \(x_{EE}\) cannot undergo a Hopf bifurcation. Indeed, the necessary condition for a Hopf bifurcation is \(T=0\), which holds if and only if \(AB-C=0\). However,
\begin{eqnarray*}
AB-C & = & (s_1 + s_2 + s_3)\big[s_1(s_2 + s_3) + (s_2 s_3 - \Delta_1) + \Delta_2 \big] \\
& & - \big[s_1(s_2 s_3 - \Delta_1) + s_2 \Delta_2\big] \\
& = & s_1(s_2 + s_3)(s_1 + s_2 + s_3) + (s_2 + s_3)(s_2 s_3 - \Delta_1) + (s_1 + s_3)\Delta_2\\
&=& > 0.
\end{eqnarray*}

Hence, the endemic equilibrium \(x_{EE}\) is locally asymptotically stable if and only if \(\mathcal{R}^2 > 1-p_v\), a condition equivalent to \(\mathcal{R}_0 > 1\).
\end{proof}

Note that $\mathcal{R}_0$ does not depend on $\epsilon_h$ or $\omega_h$. The disease-free equilibrium exists only if $\epsilon_h = 0$, while the stability of the endemic equilibrium is unaffected for small values of $\epsilon_h$ as shown below.

\begin{theorem}[Stability of the Endemic Equilibrium under Perturbations]
\label{theo:eestar}
The endemic equilibrium~\eqref{eqn:ee} remains stable under small constant additive perturbations to the host force of infection,
\[
\lambda_h = \beta_h \tfrac{I_v}{N_h} + \epsilon_h,
\]
provided $\mathcal{R}_0 > 1$.
\end{theorem}

\begin{proof}
Assume $\mathcal{R}_0 > 1$ and let 
\[
x_{EE}^\star = \left(S_h^\star, I_h^\star, R_h^\star, S_v^\star, I_v^\star\right)
\]
denote the endemic equilibrium of system~\eqref{eqn:vectorhost} when $\epsilon_h > 0$ is small. Its components satisfy
\begin{eqnarray*}
I_v^\star & = & \dfrac{I_v^{EE}}{2} - \frac{\mathcal{R}_v N_v}{2\mathcal{R}^2}\dfrac{\epsilon_h}{\mu_h + \gamma_h} \\
& & + \sqrt{\left(\dfrac{I_v^{EE}}{2} - \frac{\mathcal{R}_v N_v}{2\mathcal{R}^2}\dfrac{\epsilon_h}{\mu_h + \gamma_h}\right)^2 
+ \left(\frac{\mathcal{R}_v N_v}{\mathcal{R}^2}\right)\dfrac{\mathcal{R}_v N_v}{\mathcal{R}_v + \left(1-p_v\right)\left(1+Q\right)}\dfrac{\epsilon_h}{\mu_h + \gamma_h} }, \\
S_v^\star & = & N_v - I_v^\star, \\
I_h^\star & = & \dfrac{1-p_v}{R_v} \dfrac{I_v^\star}{S_v^\star}N_h, \\
R_h^\star & = & Q I_h^\star, \\
S_h^\star & = & N_h - \left(1+Q\right)I_h^\star. 
\end{eqnarray*}

To study how $I_v^\star$ depends on $\epsilon_h$, consider the derivative at $\epsilon_h=0$:
\[
\left.\dfrac{dI_v^\star}{d\epsilon_h}\right|_{\epsilon_h = 0} 
= N_v \dfrac{\mathcal{R}_v}{\mathcal{R}^2}\dfrac{1-p_v}{\mathcal{R}^2 - \left(1-p_v\right)}\dfrac{1}{\mu_h + \gamma_h} > 0.
\]
Thus, $I_v^\star$ increases with $\epsilon_h$, implying $I_v^\star \geq I_v^{EE}$ for $\epsilon_h > 0$. 

By the identity 
\[
I_h^\star  =  \dfrac{1-p_v}{\mathcal{R}_v} \dfrac{I_v^\star}{S_v^\star}N_h,
\]
$I_h^\star$ is also an increasing function of $\epsilon_h$ near zero. Ecologically, small perturbations $\epsilon_h$ raise the infected populations ($I_h^\star$ and $I_v^\star$) and reduce the susceptible populations ($S_h^\star$ and $S_v^\star$). Formally, this yields the inequalities
\begin{equation}\label{eq:perturbed-inequalities}
I_h^\star \geq I_h^{EE}, \quad I_v^\star \geq I_v^{EE}, \quad 
S_h^\star \leq S_h^{EE}, \quad S_v^\star \leq S_v^{EE}.
\end{equation}

As in the proof of Theorem~\ref{theo:ee}, stability of $x_{EE}^\star$ follows if 
$\tilde{s}_2 \tilde{s}_3 - \tilde{\Delta}_1 >0$, where
\[
\tilde{s}_2 = \left(1-p_v\right)\mu_v + \beta_v \dfrac{I_h^\star}{N_h}, \qquad 
\tilde{s}_3 = \left(\mu_h + \gamma_h + \epsilon_h\right) + \beta_h \dfrac{I_v^\star}{N_h},
\]
\[
\tilde{\Delta}_1 = \left(\beta_h \dfrac{S_h^\star}{N_h}\right)\left(\beta_v \dfrac{S_v^\star}{N_h}\right).
\]
From inequalities~\eqref{eq:perturbed-inequalities} it follows that 
$\tilde{s}_2 \geq s_2$, $\tilde{s}_3 \geq s_3$, and $\tilde{\Delta}_1 < \Delta_1$. Therefore,
\[
\tilde{s}_2 \tilde{s}_3 - \tilde{\Delta}_1 \geq s_2 s_3 - \Delta_1 > 0,
\]
which shows that the characteristic polynomial of the perturbed equilibrium $x_{EE}^\star$ has no positive real roots. Consequently, the perturbed endemic equilibrium $x_{EE}^\star$ remains locally stable and does not undergo a Hopf bifurcation.
\end{proof}

\subsection{Sobol sensitivity analysis}
\label{subsec:sobol}

We investigate how variations in the parameters of the macroscopic model~\eqref{eqn:vectorhost} influence key epidemiological characteristics of dengue dynamics. To characterize transient epidemic behavior, we compute the maximum, minimum, standard deviation, and area under the curve (AUC) of the host and vector trajectories. These indicators quantify outbreak magnitude and variability, enabling us to assess the contribution of each parameter. We then combine these measures with sensitivity analysis to rank parameters by relative importance across the ranges specified in Table~\ref{tab:parameters}.

\begin{remark}[Normalization by $N_h$]
Because the human population size $N_h$ is constant, dividing system~\eqref{eqn:vectorhost} by $N_h$ yields equations in terms of the proportions $s_h$, $i_h$, and $r_h$, together with the ratio $C=N_v/N_h$. The normalized model retains the same bilinear incidence structure, satisfies $s_h+i_h+r_h=1$, and preserves positivity, the invariant region, and threshold quantities such as $\mathcal{R}_0$. Thus, normalization by $N_h$ merely rescales the variables without altering the qualitative dynamics. In practice, the normalized formulation is also more convenient for numerical simulation.
\end{remark}

We divide the system of equations~\eqref{eqn:vectorhost} by $N_h$ and consider the input parameter vector
\[
X = (C, \mu_h, \mu_v, \gamma_h, \beta_h, \beta_v, \omega_h, \epsilon_h, p_v)^T,
\]
where $C=\tfrac{N_v}{N_h}\in[10^{-1},10^{4}]$, and define the model output as
\[
Y = \text{op}_{t \in [0,T]} \, (S_h(t), I_h(t), R_h(t), S_v(t), I_v(t))^T \in \mathbb{R}^5,
\]
where the operator \(\text{op} \in \{\max, \min, \mathrm{std}, \mathrm{auc}\}\) applies the chosen statistic over the time interval \([0,365]\) to each state variable.

To assess parameter influence, we use the \textit{multivariate first-order Sobol sensitivity index} and the \textit{total effect index} introduced by Saltelli \textit{et al.}~\cite{saltelli2010variance}. Let
\[
Y = f(X_1, X_2, \ldots, X_k) \in \mathbb{R}^d
\]
denote the model output vector for \(k\) independent inputs. The first-order index for input \(X_i\) is defined as
\begin{equation}
S_i = \frac{\mathrm{Tr}\!\left( \mathrm{Var}_{X_i}\!\left( \mathbb{E}_{X_{\sim i}}[Y \mid X_i] \right) \right)}{\mathrm{Tr}\!\left( \mathrm{Var}(Y) \right)},
\end{equation}
where \(X_{\sim i}\) denotes all inputs except \(X_i\), and \(\mathrm{Tr}(\cdot)\) sums variances across output components. Thus, \(S_i\) measures the proportion of output variance explained solely by \(X_i\).

The total effect index is
\begin{equation}
S_{T_i} = \frac{\mathrm{Tr}\!\left( \mathbb{E}_{X_{\sim i}}  \!\left(  \mathrm{Var}_{X_i}  [Y \mid X_{\sim i}] \right) \right)}{\mathrm{Tr}\!\left( \mathrm{Var}(Y) \right)}
= 1-\frac{\mathrm{Tr}\!\left( \mathrm{Var}_{X_{\sim i}}\!\left( \mathbb{E}_{X}[Y \mid X_{\sim i}] \right) \right)}{\mathrm{Tr}\!\left( \mathrm{Var}(Y) \right)},
\end{equation}
which captures the share of variance that disappears if \(X_i\) is fixed, accounting for both its direct effect and interactions.
 Numerical experiments are implemented with the \textit{SALib} Python package of Herman and Usher~\cite{Herman2017}. To approximate the onset of local transmission, the initial conditions are specified as
\begin{equation}
\label{eq:ic_dengue_alt}
\begin{aligned}
(S_h(0), I_h(0), R_h(0), S_v(0), I_v(0)) 
= \bigl(& (1-10^{-5})N_h,\; 10^{-5}N_h,\; 0,\\
        & (1-10^{-4})C N_h,\; 10^{-4}C N_h \bigr),
\end{aligned}
\end{equation}
which corresponds to introducing a small fraction of infectious individuals in both the human and vector populations.

Parameter values are drawn uniformly from the biologically feasible ranges summarized in Table~\ref{tab:model}. Host and vector population sizes (\(N_h, N_v\)) are varied according to a proportion $C$, while epidemiological parameters such as transmission rates (\(\beta_h, \beta_v\)), recovery rate (\(\gamma_h\)), waning immunity (\(\omega_h\)), exogenous infection rate (\(\epsilon_h\)), and vertical transmission probability (\(p_v\)) are varied independently. 

Although the parameter ranges are normalized to published values rather than estimated from incidence data, this framework enables the exploration of a wide spectrum of dynamical behaviors under consistent assumptions. It thus provides a systematic way to identify the parameters that have the greatest impact on dengue outbreak statistics.

The sensitivity results shown in Figures~\ref{fig:parameter_importance} and~\ref{fig:parameter_ranking} illustrate both the overall importance of the parameters and their contributions to specific quantities of interest. Notably, the host recovery rate $\gamma_h$ has a strong influence on the area under the curve and on the standard deviation of the number of infected humans. This underscores the relevance of designing control campaigns aimed at protecting infectious hosts from mosquito bites to control dengue.

\subsection{Variance of the endemic equilibrium}
\label{subsec:psd}

In order to carry out the uncertainty analysis of the endemic equilibrium~\eqref{eqn:ee} we describe dengue transmission dynamics defined by Figure~\ref{fig:model} and Table~\ref{tab:model} using a
continuous-time Markov jump process and adapt a series of known results to the model. Let the state of the system at time $t$ be denoted by a vector of random variables $X(t) = (X_1(t), X_2(t), X_3(t), X_4(t), X_5(t))$, where $X_1$, $X_2$, and $X_3$ represent the number of susceptible ($S_h$), infectious ($I_h$), and recovered ($R_h$) humans, respectively, and $X_4$ and $X_5$ denote the number of susceptible ($S_v$) and infectious ($I_v$) mosquitoes. We assume that the total human and vector populations remain constant over time, i.e., $X_1(t) + X_2(t) + X_3(t) = N_h$ and $X_4(t) + X_5(t) = N_v$. 
We denote a realization of the vector $X$ with lowercase $x(t) = (x_1(t), x_2(t), x_3(t), x_4(t), x_5(t))$.

The possible epidemic events, along with their corresponding propensities and stoichiometry, are listed in Table~\ref{tab:model}. Infection rates for both hosts and vectors are scaled by the total human population $N_h$. 

The joint probability distribution of the state variables is governed by a Forward Kolmogorov equation
\begin{equation}
\label{eq:cma}
\frac{dP(x,t)}{dt}  = \sum_{j=1}^{12}\alpha_{j}(x-v_{j})P(x-v_{j},t)-\sum_{j=1}^{12}\alpha_{j}(x)P(x,t).
\end{equation}

Under the random mixing hypothesis, it is standard to represent the state variables $x_i$ as 
\[
\frac{x_{i}}{\Omega} = y_{i}+\frac{z_{i}}{\sqrt{\Omega}},
\]
where $y_{i}$ denotes the deterministic mean, $z_{i}$ the stochastic fluctuation, and $\Omega$ the system size. 

For vector-borne epidemic models, however, the scaling must reflect host and vector population sizes. In particular, host variables scale with the host population $N_h$, while vector variables scale with the vector population $N_v$:
\begin{equation}
\label{eq:linear_noise}
\begin{split}
\dfrac{x_{i}}{N_h} &= y_{i}+\dfrac{z_{i}}{\sqrt{N_h}},\quad i=1,2,3,\\
\dfrac{x_{i}}{N_v} &= y_{i}+\dfrac{z_{i}}{\sqrt{N_v}},\quad i=4,5.
\end{split}
\end{equation}

We take the host population size $N_h=\Omega$ as the system size and define the vector-to-host ratio $C=N_v/N_h$. Consequently, we write
\begin{equation}
\label{eq:scale}
\dfrac{x_{i}}{N_h} = \varphi_{i}+\dfrac{\psi_{i}}{\sqrt{N_h}},\quad i=1,\ldots,5,
\end{equation}
where
\[
\varphi_{i} =
\begin{cases}
y_{i}, & i=1,2,3,\\
C y_{i}, & i=4,5,
\end{cases}
\qquad
\psi_{i} =
\begin{cases}
z_{i}, & i=1,2,3,\\
\sqrt{C}\,z_{i}, & i=4,5.
\end{cases}
\]

We compute the power spectral density (PSD) of fluctuations in the stochastic epidemic model. For the van Kampen system-size expansion, we adopt the notation of Grima~\cite{grima2010effective} and Thomas {\it et al.}~\cite{thomas2012intrinsic}, and we follow Komorowski {\it et al.}~\cite{komorowski2009bayesian} to write the macroscopic system. The deterministic mean-field dynamics are
\begin{equation}
\label{eq:mean_field}
\frac{d\varphi_{i}}{dt}=\sum_{k=1}^{12}S_{ik}\,\alpha_{k}(\varphi),
\end{equation}
which describe the evolution of the rescaled state variables. This is the same equation as the deterministic model~\eqref{eqn:vectorhost}.

In the linear noise approximation, The joint probability distribution $\Pi(\psi)$ of the system fluctuations is the following Fokker-Planck equation
\begin{equation}
\label{eqn:fokker_planck}
\frac{\partial\Pi(\varphi,\psi,t)}{\partial t} = 
-\sum_{i,j=1}^{5}A_{ij}(\varphi)\frac{\partial}{\partial \psi_{j}}\big(\psi_{i}\Pi\big)
+\frac{1}{2}\sum_{i,j=1}^{5}B_{ij}(\varphi)\frac{\partial^{2}\Pi}{\partial \psi_{i}\,\partial \psi_{j}},
\end{equation}
where the drift and diffusion matrices are given by
\begin{equation}
\label{eqn:drift_diffusion}
A_{ij}(\varphi)=\sum_{k=1}^{12}S_{ik}\,\dfrac{\partial\alpha_{k}(\varphi)}{\partial\varphi_{j}},
\qquad
B_{ij}(\varphi)=\sum_{k=1}^{12}S_{ik}S_{jk}\,\alpha_{k}(\varphi).
\end{equation}

Following Alonso \textit{et al.}~\cite{alonso2007stochastic}, we apply Fourier theory to study the variance of the endemic equilibrium $\varphi^{EE}$ defined in equation~\eqref{eqn:ee}. For this purpose equation~\eqref{eqn:fokker_planck} is approximated by the linear Langevin equation
\begin{equation}
\label{eqn:langevin}
\frac{d\psi}{dt} = A(\varphi^{EE})\psi+\eta(t),
\end{equation}
where $\mathbb{E}[\eta(t)]=0$ and $\mathbb{E}[\eta(t)\eta(s)] = B(\varphi^{EE})\delta(t-s)$. 
Särkkä and Solin~\cite{sarkka2019applied} derive the power spectral density
\[
S(\omega,\varphi^{EE}) \;=\; (i\omega I - A(\varphi^{EE}))^{-1} \, B(\varphi^{EE}) \, (-i\omega I - A(\varphi^{EE})^{\top})^{-1},
\]
and show that the stationary covariance $\Gamma(\varphi^{EE})$ satisfies the Lyapunov equation
\begin{equation}
\label{eqn:lyapunov}
A(\varphi^{EE}) \Gamma + \Gamma(\varphi^{EE}) A(\varphi^{EE})^{\top} + B(\varphi^{EE}) = 0.
\end{equation}
Moreover, by the Wiener--Khinchin theorem, the stationary covariance $\Gamma(\varphi^{EE})$ given by equation~\eqref{eqn:lyapunov} equals the inverse Fourier transform of the power spectral density
\[
\Gamma(\varphi^{EE})=\mathcal{F}^{-1}\{S(\omega,\varphi^{EE})\}.
\] 
In Section~\ref{sec:results}, we solve equation~\eqref{eqn:lyapunov} numerically to analyze $\Gamma(\varphi^{EE})$ in the contact-rate plane.

\begin{remark}
It is straightforward to compute the covariance of the endemic equilibrium for any epidemic model with vital dynamics defined from first principles similar to Figure~\ref{fig:model} and Table~\ref{tab:model}.
\end{remark}

\section{Results}
\label{sec:results}

In this Section we gather numerical evidence which relies on the mathematical foundation provided by the qualitative analysis of the macroscopic model developed in Section~\ref{sec:macroscopic}. Section~\ref{sec:SA_results} shows the Sobol sensitivity analysis. Section~\ref{sec:PSD_results} presentes the uncertainty analysis of the endemic equilibrium.

\subsection{Sensitivity analysis} 
\label{sec:SA_results}
In the setting described in Subsection~\ref{subsec:sobol}, we rank dengue parameters in the context of a seasonal outbreak. Figure~\ref{fig:parameter_importance} shows the mean sensitivity of the model parameters across all outputs. The quotient of population sizes $C$ is the leading parameter in importance in the macroscopic model during an outbreak. In terms of the population fractions $i_h=\tfrac{I_h}{N_h}$ and $i_v=\tfrac{I_v}{N_h}$, the forces of infection are
\[
\lambda_h = C\beta_h i_h+\epsilon_h,\qquad \lambda_v = \beta_v i_v,
\]
showing that $C$ affects transmission in both directions asymmetrically. This is confirmed in the first column of Figure~\ref{fig:parameter_ranking}. Because all considered summary statistics are monotonic with respect to the force of infection, uncertainty in $C$ is propagated to all four outputs asymmetrically, affecting more the vector.

The condition \(ST(C) > S1(C)\) implies strong interactions. Parameter $C$ acts together with \(\beta_h\) and \(\beta_v\), so uncertainty in \(C\) also increases the apparent effect of other parameters, e.g. $\mathcal{R}^{2}$. Two practical public health recommendations follow immediately. First, obtain reliable estimates of \(C\). Although the vector--to--host ratio \(C\) is not directly measurable, it can be inferred using standard entomological indicators of mosquito abundance---such as adult mosquito trap counts~\cite{Narayanasamy2025}, pupal/productivity surveys~\cite{focks2006multicountry}, and mark--release--recapture studies~\cite{guerra2014global}---combined with demographic estimates of the human population. Calibration models linking these indices to adult mosquito density are increasingly used in dengue surveillance~\cite{polwiang2020time,aswi2019bayesian}. Our sensitivity analysis indicates that uncertainty in these estimates can dominate epidemic predictions, underscoring the importance of improving the reliability of \(C\) for both inference and control.


Second, implement measures that reduce \(C\) (larval control, household screening, and safer water storage), which should reduce peak size, AUC, and standard deviation. The caveat is that poorly measured \(C\) can obscure the benefits of improving contact-rate estimates and forecasts. For surveillance, priority should be given to robust proxies of $N_v$, such as ovitrap indices, paired with calibration models to convert indices into abundance. $N_h$ can be obtained from census and mobility-adjusted denominators. Better $C$ data improve inference and forecasting. Notably, this result is in agreement with Smith \textit{et al.}~\cite{smith2012ross} who recommend interventions informed with the vector capacity, which depends linearly on $C$.

The human recovery rate \(\gamma_h\) emerges as one of the most influential parameters in the system. In the sensitivity analysis, \(\gamma_h\) exhibits a larger first--order sensitivity index than the contact rates \(\beta_h\) and \(\beta_v\), indicating a strong direct contribution to outbreak dynamics. Mechanistically, \(\gamma_h\) controls the infectious period in humans: increasing \(\gamma_h\) shortens the time during which infected individuals contribute to transmission, thereby reducing epidemic intensity. As a result, larger values of \(\gamma_h\) are associated with lower epidemic peaks, reduced area under the incidence curve, and decreased variance of the human infectious class \(I_h\). This effect is most pronounced for the area under the curve and the standard deviation of \(I_h\), as shown in Figure~\ref{fig:parameter_ranking}.

In contrast, although the human--to--vector contact rate \(\beta_v\) has a smaller first--order sensitivity index than \(\gamma_h\), it displays a comparatively larger total sensitivity index, reflecting strong interaction effects with other model parameters. The influence of \(\beta_v\) is distributed across most epidemic summary statistics, with the exception of the minima of \(I_h\) and \(I_v\). Notably, the impact of \(\gamma_h\) on mosquito--related summary statistics is limited, whereas \(\beta_v\) plays a broader role in shaping both human and vector dynamics through parameter interactions.

Taken together, these results highlight a clear separation between parameters with dominant direct effects, such as \(\gamma_h\), and parameters whose influence arises primarily through interactions, such as \(\beta_v\). From an inference perspective, the strong sensitivity of epidemic outcomes to \(\gamma_h\) implies that uncertainty in the human recovery period can substantially affect predictions of outbreak magnitude and variability.

A nonzero $p_v$ produces infectious mosquitoes even when human prevalence is low, lowering the threshold for persistence and allowing transmission outside the main season. This appears as a higher AUC of infectious mosquitoes in Figure~\ref{fig:parameter_ranking}. This finding supports combining adult mosquito control with actions that target immature stages, such as long-lasting larvicides and container management, especially in places where field data show that $p_v$ is relevant.

Finally, on the time scale of weeks to months, human mortality $\mu_h$ has almost no effect on outbreak dynamics. Exogenous infections $\epsilon_h$ also play a minor role once local spread is established, since peak size and AUC are mainly determined by the vector--host ratio $C$ and the human recovery rate $\gamma_h$, with the exception of highly connected cities at the start of the season. On the other hand, $\epsilon_h$ affects endemicity by increasing the infectious endemic states as shown in the proof of Theorem~\ref{theo:eestar}. Vector mortality $\mu_v$ has only a small impact on outbreak size, AUC, and variability over short periods, because its effect is partly balanced by $C$ and the transmission rates $\beta_h,\beta_v$, and because neither strong seasonality nor incubation dynamics are included in the time frame considered. Waning immunity $\omega_h$ mainly influences multi-season behavior and long-term persistence, but its effect is weak for a single-season deterministic outbreak.

\begin{figure}[ht]
\centering
\includegraphics[width=\textwidth]{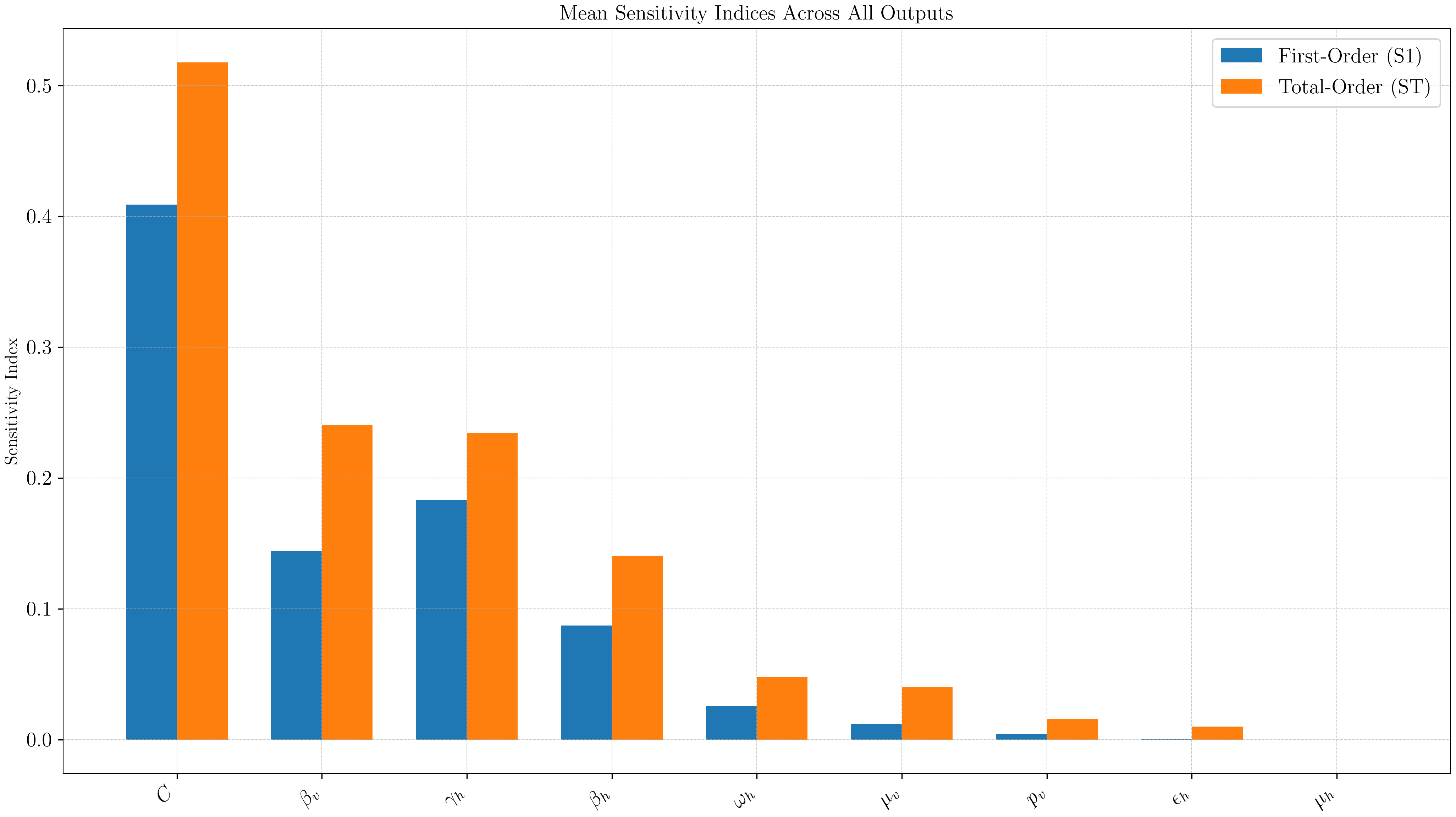}
\caption{
Mean first--order ($S_1$, blue) and total--order ($S_T$, orange) Sobol sensitivity indices of the dengue model parameters, averaged across outbreak summary statistics (peak incidence, area under the curve, standard deviation, and minimum incidence). The vector--to--host ratio $C$ is the dominant parameter, exhibiting the largest total--order sensitivity index and indicating a strong influence that includes interaction effects with other parameters. Among the remaining parameters, the human recovery rate $\gamma_h$ has the largest first--order sensitivity index, highlighting its dominant direct effect on outbreak magnitude and variability. In contrast, the mosquito--to--human transmission rate $\beta_v$ displays a comparatively larger total--order index than first--order index, indicating that its influence arises primarily through interactions rather than direct effects. Vertical transmission $p_v$ has a moderate but non--negligible contribution, while exogenous infection $\epsilon_h$ and human mortality $\mu_h$ have negligible influence on single--season outbreak dynamics. Vector mortality $\mu_v$ and waning immunity $\omega_h$ show only minor effects within the considered time scale.
}
\label{fig:parameter_importance}
\end{figure}

\begin{figure}[ht]
\centering
\includegraphics[width=0.9\textwidth]{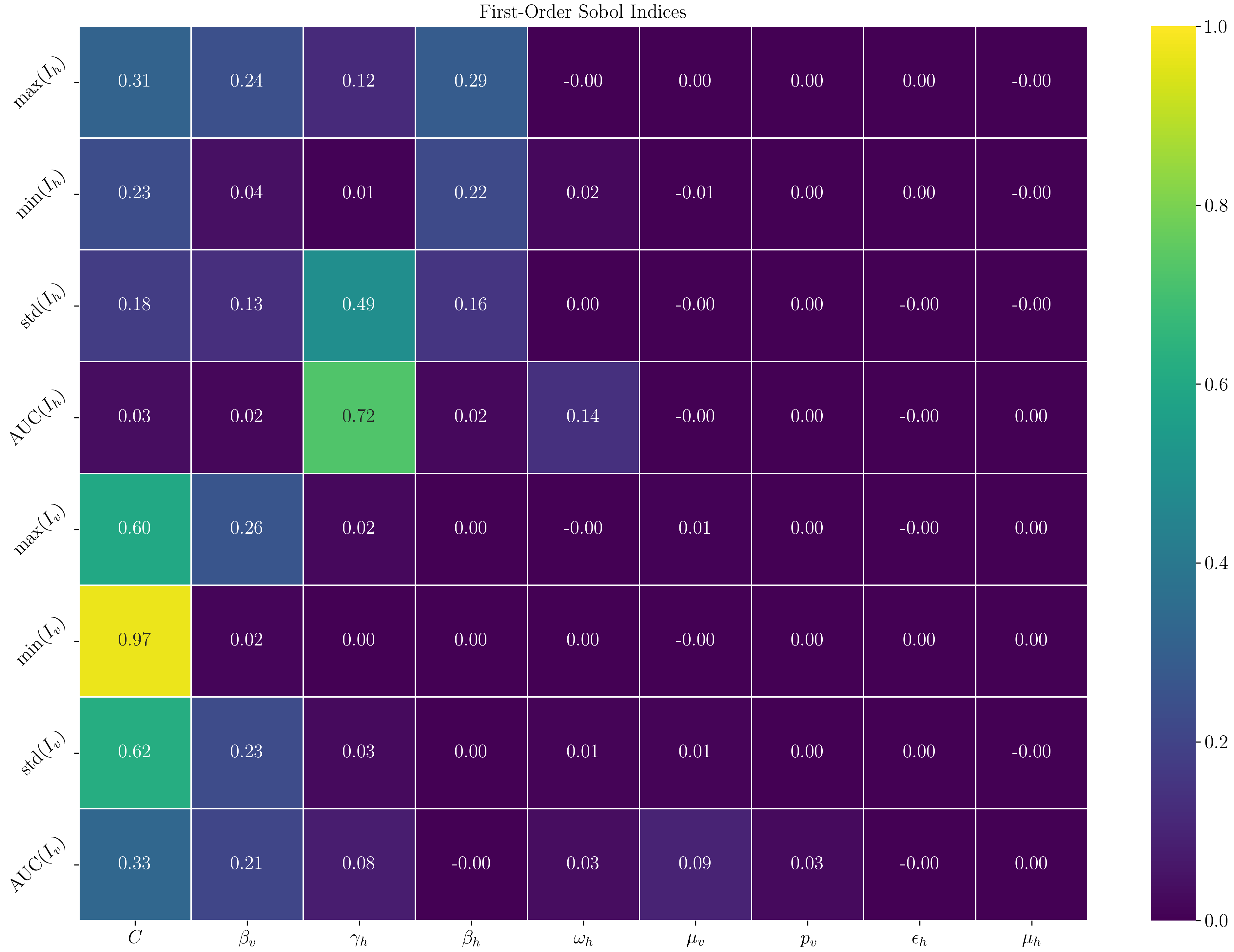}
\caption{
First-order Sobol indices ($S1$) of dengue model parameters are shown for each outbreak statistic: maximum, minimum, standard deviation, and area under the curve (AUC) of human ($I_h$) and vector ($I_v$) infectious fractions. The vector-to-host ratio $C$ strongly affects maxima and minima in both populations, while the human recovery rate $\gamma_h$ dominates the variance and AUC of human infections. Vertical transmission $p_v$ mainly influences the mosquito AUC, and $\beta_v$ contributes to variability in both hosts and vectors. Human mortality $\mu_h$, exogenous infection $\epsilon_h$, and waning immunity $\omega_h$ have negligible first-order effects in the single-season deterministic setting. This decomposition clarifies parameter-specific roles that aggregate rankings may obscure.
}
\label{fig:parameter_ranking}
\end{figure}

\subsection{Uncertainty analysis of the endemic equilibrium}
\label{sec:PSD_results}

The comparison between the variances of the infectious classes and the basic reproductive number $\mathcal{R}_0$ as functions of the contact rates $\beta_{h}$ and $\beta_{v}$ in Figure~\ref{fig:psd} reveals a fundamental difference between deterministic and stochastic perspectives of disease dynamics. The contour plot of $\mathcal{R}_0$ shows a smooth and monotonic increase as either contact rate grows, reflecting the well-known threshold property: once $\mathcal{R}_0 > 1$, endemic transmission becomes sustainable, and larger contact rates reinforce persistence. The line $\mathcal{R}_0=1$ in the figure is determined by the population ratio $C=N_{v}/N_{h}$: for larger $C$, the threshold can be crossed with smaller values of both $\beta_{h}$ and $\beta_{v}$. Consequently, large $C$ shifts the system into the region of maximal stochastic fluctuations more easily, underscoring the need to estimate this ratio carefully and prioritize interventions that reduce it. The variance plots further illustrate this point. For humans ($\mathrm{var}(I_h^{EE})$), fluctuations are concentrated in a narrow region near low $\beta_v$, and decay rapidly as transmission intensifies, indicating that once mosquito-to-human transmission is high, the endemic equilibrium suppresses noise in the human class even as $\mathcal{R}_0$ continues to rise. In contrast, the variance of the vector class ($\mathrm{var}(I_v^{EE})$) displays a highly non-monotonic pattern: negligible for most parameters but extremely amplified in localized regions of the $(\beta_h, \beta_v)$ plane, with values orders of magnitude larger than those of the human class. These ``hot spots'' of stochastic amplification emerge despite smooth changes in $\mathcal{R}_0$, showing that noise intensity is not proportional to the deterministic reproductive number. In particular, large $\beta_{v}$ not only amplifies fluctuations in infectious mosquitoes but also implies the production of more infectious eggs that may hatch in the following season, reinforcing long-term persistence. Taken together, these results strengthen the policy recommendations already highlighted in this work: estimate and reduce $C$ as a primary objective, and in parallel, diminish $\beta_{v}$ by protecting infectious hosts from mosquito bites, since both mechanisms mitigate the regions where stochastic amplification could destabilize epidemiological outcomes.

\begin{figure}[ht]
\centering
\includegraphics[width=0.9\textwidth]{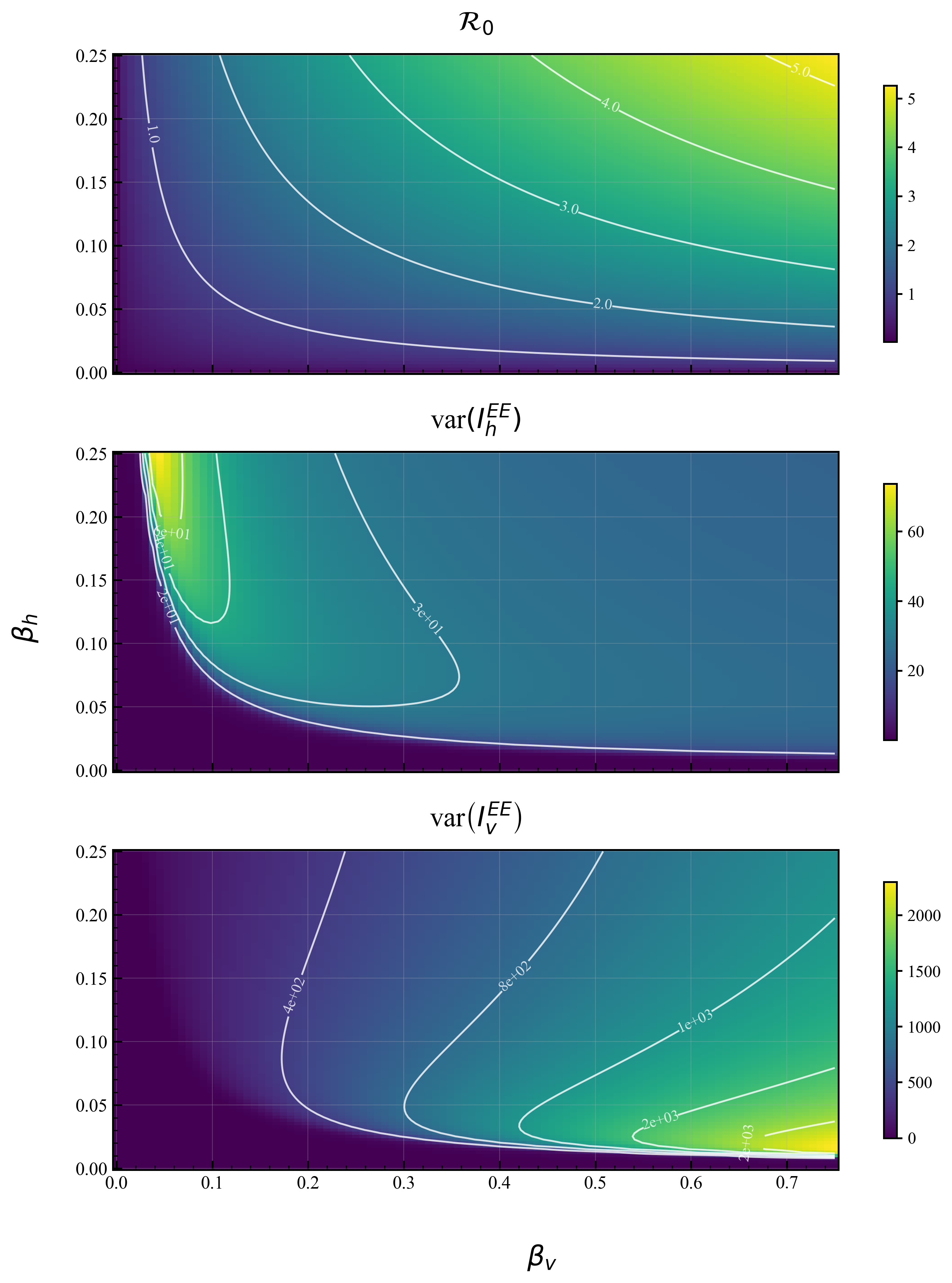}
\caption{Basic reproduction number and variance of infectious classes at the endemic equilibrium as functions of the human–to–vector ($\beta_h$) and vector–to–human ($\beta_v$) contact rates. (Top) Contour plot of the basic reproduction number showing the threshold $\mathcal{R}_0=1$ and higher values as $\beta_h$ and $\beta_v$ increase. (Middle) Variance of the number of infectious humans at endemic equilibrium, highlighting a region of elevated fluctuations near the threshold. (Bottom) Variance of infectious vectors at endemic equilibrium, with large fluctuations emerging for high $\beta_v$ values. Together, the panels illustrate how transmission heterogeneity in contact rates shapes both the deterministic threshold for persistence and the intensity of stochastic fluctuations.}
\label{fig:psd}
\end{figure}

To complement the variance maps and provide direct intuition for the localized variance hot spots, we performed stochastic simulations of the underlying Markov jump process using a Gillespie algorithm, see Figure~\ref{fig:gillespie}. All simulations were conducted while holding the deterministic reproduction number fixed at $\mathcal R_0 = 1.3$, with vertical transmission set to $p_v = 0.05$ and the vector--to--host ratio fixed at $C = 1$. We chose $C=1$ for the sake of clarity. Larger values of $C$ preserve the patterns shown in Figure 4, but the line $\mathcal{R}_{0}=1$ gets closer to the $\beta_v$--$\beta_h$ axes. Under these constraints, the human--to--mosquito and mosquito--to--human contact rates satisfy
\[
\beta_h \beta_v
=
\frac{\left(\left(\mathcal{R}_{0}-\frac{p_v}{2}\right)^{2}-\left(\frac{p_v}{2}\right)^{2}\right)^{2}}{C}
(\gamma_h+\mu_h)\mu_v,
\]
which defines a curve in the $(\beta_h,\beta_v)$ plane along which $\mathcal R_0$ remains constant.

We selected representative parameter pairs lying on this fixed--$\mathcal R_0$ curve but belonging to distinct variance regimes identified in Fig.~4. In the low--variance regime, stochastic trajectories of both infected humans and infected vectors remain tightly clustered around the median, indicating weak stochastic amplification. In contrast, in the high--variance regime, mosquito infection trajectories exhibit pronounced realization--to--realization variability, characterized by intermittent burst--like dynamics, despite identical deterministic transmission potential. These simulations demonstrate that the variance hot spots arise from how the fixed transmission potential encoded in $\mathcal R_0$ is distributed between $\beta_h$ and $\beta_v$, rather than from changes in $\mathcal R_0$ itself. As a result, stochastic amplification can dominate epidemic variability in localized regions of parameter space even when the deterministic threshold is held constant.

\begin{figure}[ht]
\centering
\includegraphics[width=\textwidth]{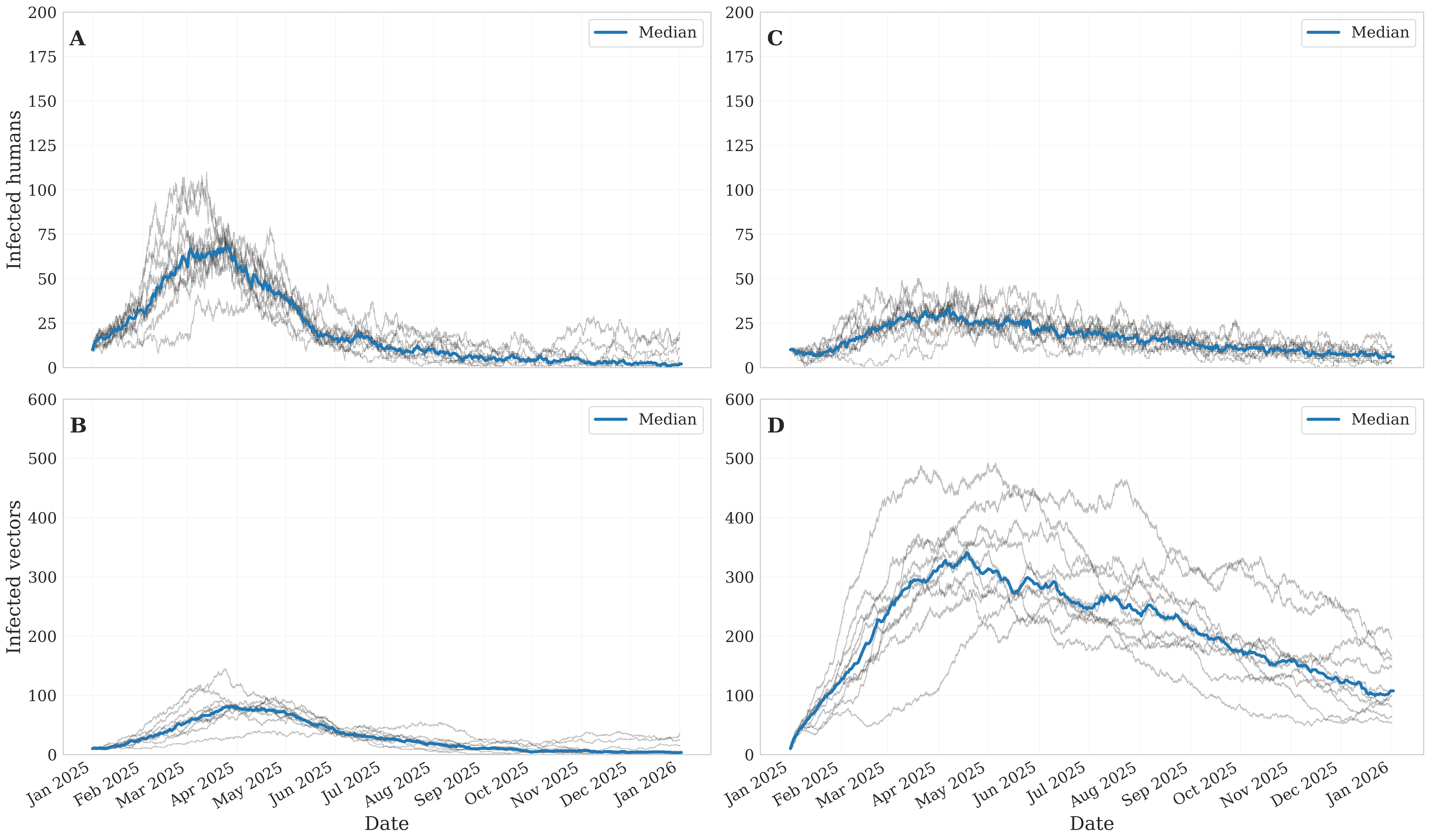}
\caption{Gillespie realizations of the dengue transmission model for two parameter pairs lying on the same fixed--$\mathcal R_0$ curve ($\mathcal R_0 = 1.3$, $p_v = 0.05$, $C = 1$). Left panels: low--variance regime ($\beta_h = 0.25$, $\beta_v = 0.06$), where stochastic trajectories of infected humans and vectors remain tightly clustered around the median. Right panels: high--variance regime ($\beta_h = 0.015$, $\beta_v = 0.75$), where mosquito infection dynamics exhibit strong realization--to--realization variability and intermittent burst--like behavior despite identical deterministic reproduction numbers. Gray lines denote individual stochastic trajectories; the blue line indicates the median trajectory across realizations.
}
\label{fig:gillespie}
\end{figure}

\section{Discussion}
\label{sec:discussion}

The modeling, qualitative and quantitative analysis presented in this work extends the Dietz model of mosquito-borne pathogen transmission (as reviewed by Smith \textit{et al.}~\cite{smith2012ross}) in several novel directions, with particular emphasis on dengue dynamics and control. First, by combining the linear noise approximation with Wiener-Khinchin theorem, we compute the stationary covariance of the endemic equilibrium and map the variance of the infectious human and vector classes across the $(\beta_h,\beta_v)$ plane. These maps reveal localized regions of strong amplification, which stand in contrast to the smooth and monotonic dependence of $\mathcal{R}_0$ on the same parameters. Importantly, the maxima of human and vector variance are in different parts of the contact rate plane. The asynchronous location of these two regions suggests an ecological and evolutionary advantage for dengue: the pathogen effectively has two distinct niches to occupy during the endemic phase, one dominated by fluctuations in humans and another by fluctuations in mosquitoes. This observation generalizes naturally to the multi-strain case, where different viral lineages may exploit separate fluctuation regimes to persist in the host--vector system.

Second, we perform a global Sobol sensitivity analysis of epidemic summary statistics, including peak prevalence, trough values, variance, and area under the curve. This analysis establishes a clear ranking of parameter influence, with the vector--to--host ratio $C=N_v/N_h$ and the human recovery rate $\gamma_h$ emerging as the most influential drivers of seasonal dengue dynamics, surpassing the contact rates themselves. This finding leads directly to the recommendation that $C$ be estimated carefully in field studies and that reducing $C$ should be a central control objective. Mechanistically, a large vector population (\(C\) large) increases the absolute number of vertically infected offspring produced per unit time when \(p_v>0\). When mosquito populations experience strong stochastic fluctuations during the off--season, this enlarged pool of infected eggs can persist through unfavorable conditions and subsequently hatch as infectious mosquitoes at the onset of the next transmission season. This mechanism provides a biologically plausible pathway by which vertical transmission, in combination with large vector populations and demographic noise, may contribute to the persistence and re--emergence of dengue across transmission seasons.


Third, we show that large values of $\beta_v$ not only increase the variance of infectious mosquitoes but also imply the production of more infectious eggs that can hatch in subsequent seasons. This coupling provides a mechanistic link between stochastic amplification and inter-seasonal persistence, and it reinforces the need to protect infectious hosts from mosquito bites in order to reduce $\beta_v$ and its long-term consequences.

Together, these results refine the Dietz framework by adding a stochastic dimension, quantifying parameter influence through global sensitivity analysis, and linking these insights to concrete public health priorities: estimate and reduce $C$ as a first step, and reduce $\beta_v$ by protecting infectious hosts from mosquito bites.

\section{Conclusions}
\label{sec:conclusions}

In conclusion, our analysis builds on the Dietz framework by incorporating stochastic fluctuations, global sensitivity analysis, and extended threshold formulae to better understand dengue dynamics and control. The results demonstrate that variance patterns in the infectious human and vector classes are not simply proportional to $\mathcal{R}_0$, but instead display localized regions of strong amplification that reveal hidden ecological structure in the contact-rate plane. The asynchrony of these regions implies that dengue benefits from two separate fluctuation niches during the endemic phase, an advantage that may extend to multi-strain dengue by allowing different serotypes to exploit distinct ecological regimes. At the same time, the global Sobol sensitivity analysis underscores the dominant role of the vector--to--host ratio $C$ and the human recovery rate $\gamma_h$ in shaping outbreak dynamics, placing these parameters above the contact rates themselves in importance. 
In this setting, the interaction between a large vector population (\(C\) large), nonzero vertical transmission (\(p_v>0\)), and strong stochastic fluctuations in mosquito abundance increases the absolute production and survival of infected eggs. Together, these factors provide a mechanistically consistent pathway by which infectious mosquitoes may emerge at the start of subsequent transmission seasons, contributing to the persistence and re--initiation of dengue transmission.
These findings have clear policy implications: reliable field estimation of $C$ should be prioritized, followed by interventions that reduce mosquito abundance relative to the host population; simultaneously, measures that lower $\beta_v$ by protecting infectious humans from mosquito bites are critical, as they reduce both immediate fluctuations in vector infection and the reservoir of infectious eggs that can sustain transmission across seasons. Taken together, the theoretical and numerical results show that dengue persistence and variability are governed not only by deterministic thresholds but also by the structure of stochastic amplification, and that effective control must act on the ecological levels that determine both mean transmission and the magnitude of fluctuations.

\clearpage
\bibliographystyle{unsrt}
\bibliography{references.bib}

\end{document}